\documentclass[a4paper,12pt]{article}
\usepackage{a4wide}
\usepackage[usenames,dvipsnames]{color}
\usepackage{amssymb,amsmath,eepic,graphics,color,enumerate,eucal}
\usepackage{tikz}
\usetikzlibrary{calc}
\usepackage{todonotes}

\usepackage{epsfig}
\usepackage{amsmath,amsthm,amsfonts,amssymb,epsfig,color,xy}

\newtheorem{theorem}{Theorem}[]
\newtheorem{lemma}[theorem]{Lemma}
\newtheorem{corollary}[theorem]{Corollary}

\newtheorem{observation}[theorem]{Observation}
\newtheorem{remark}[theorem]{Remark}

\newcommand{\mycase}[1]{\noindent{\bf CASE #1\ }}

\newtheoremstyle{s2}{9pt}{9pt}{\rm}{\parindent}{\bf}{.}{0.5em}{}
\theoremstyle{s2} 

\newtheorem{definition}[theorem]{Definition}
\newtheorem{algorithm}{Algorithm}

\DeclareMathOperator{\ch}{ch}

\DeclareMathOperator{\con}{con}

\newcommand{\f}{\overline{\pi}}
\newcommand{\floor}[1]{\ensuremath{\left\lfloor{#1}\right\rfloor}}%
\newcommand{\ceil}[1]{\ensuremath{\left\lceil{#1}\right\rceil}}%
\newcommand{\ignore}[1]{}%

\begin{document}
\title{Beyond the Shannon's Bound}
\author{
  Micha{\l} Farnik\thanks{Theoretical Computer Science Department, Faculty of Mathematics and Computer Science, Jagiellonian University, Krak{\'o}w, Poland; email: \texttt{michal.farnik@gmail.com}}
  \and
  \L{}ukasz Kowalik\thanks{Institute of Informatics, University of Warsaw, Poland; email: \texttt{kowalik@mimuw.edu.pl}. }
  \and
  Arkadiusz Soca{\l}a\thanks{Faculty of Mathematics, Informatics and Mechanics, University of Warsaw, Poland; email: \texttt{arkadiusz.socala@students.mimuw.edu.pl}.}
}
\date{}

\maketitle

\abstract{
Let $G=(V,E)$ be a multigraph of maximum degree $\Delta$. The edges of $G$ can be colored with at most $\frac{3}{2}\Delta$ colors by Shannon's theorem. We study lower bounds on the size of subgraphs of $G$ that can be colored with $\Delta$ colors. 

Shannon's Theorem gives a bound of $\frac{\Delta}{\lfloor\frac{3}{2}\Delta\rfloor}|E|$. 
However, for $\Delta=3$, Kami\'{n}ski and Kowalik~\cite{swat,Bey} showed that there is a 3-edge-colorable subgraph of size at least $\frac{7}{9}|E|$, unless $G$ has a connected component isomorphic to $K_3+e$ (a $K_3$ with an arbitrary edge doubled).
Here we extend this line of research by showing that $G$ has a $\Delta$-edge colorable subgraph with at least $\frac{\Delta}{\lfloor\frac{3}{2}\Delta\rfloor-1}|E|$ edges, unless $\Delta$ is even and $G$ contains $\frac{\Delta}{2}K_3$ or $\Delta$ is odd and $G$ contains $\frac{\Delta-1}{2}K_3+e$. Moreover, the subgraph and its coloring can be found in polynomial time.

Our results have applications in approximation algorithms for the Maximum $k$-Edge-Colorable Subgraph problem, where given a graph $G$ (without any bound on its maximum degree or other restrictions) one has to find a $k$-edge-colorable subgraph with maximum number of edges. 
In particular, for every even $k \ge 4$ we obtain a $\frac{2k+2}{3k+2}$-approximation and for every odd $k\ge 5$ we get a $\frac{2k+1}{3k}$-approximation. When $4\le k \le 13$ this improves over earlier algorithms due to Feige et al.~\cite{FOW02}.
}

\section{Introduction}

In this paper we consider undirected multigraphs (though for simplicity we will call them graphs).
A graph is $k$-edge-colorable if there exists an assignment of $k$ colors to the edges of the graph, such that every two incident edges receive different colors. By Shannon's theorem~\cite{shannon}, $\floor{\frac{3}{2}\Delta}$ colors suffice to color any multigraph, where $\Delta$ denotes the maximum degree. 
This bound is tight, e.g.\ for every even $\Delta$ consider the graph $(\Delta/2)K_3$, and for odd $\Delta$ consider the graph $\floor{\Delta/2}K_3+e$ (see Fig.~\ref{fig:2k3} and~\ref{fig:2k3+e} and Section~\ref{sec:comb-res} for definitions). 

It is natural to ask {\em how many} edges of a graph of maximum degree $\Delta$ can be colored with {\em less} than $\floor{\frac{3}{2}\Delta}$ colors. The {\em maximum $k$-edge-colorable subgraph of $G$} (maximum $k$-ECS in short) is a $k$-edge-colorable subgraph $H$ of $G$ with maximum number of edges. Let $\gamma_k(G)$ denote the ratio $|E(H)|/|E(G)|$; when $|E(G)|=0$ we define $\gamma_k(G)=1$. If $\Delta$ is the maximum degree of $G$ we write shortly $\gamma(G)$ for $\gamma_{\Delta}(G)$.
Lower bounds on $\gamma(G)$ were studied first by Albertson and Haas~\cite{AH96}. They showed that $\gamma(G)\ge\frac{26}{31}$ for simple graphs of maximum degree 3. Today, the case of simple graphs is pretty well researched. Since by Vizing's theorem any simple graph of maximum degree $\Delta$ can be edge-colored with $\Delta+1$ colors, by simply discarding the smallest color class we get $\gamma(G)\ge\Delta/(\Delta+1)$. This ratio grows to 1 with $\Delta$, and for $\Delta \le 7$ much more precise bounds are known (see~\cite{Bey} for a discussion).

In this paper we study lower bounds on $\gamma(G)$ for {\em multigraphs}. Note that in this case we can apply Shannon's theorem similarly as Vizing's theorem above and we get the bound $\gamma(G)\ge{\Delta}/{\lfloor\frac{3}2\Delta\rfloor}$; let us call it {\em the Shannon's bound}. 
As far as we know, so far better bounds are known only for subcubic graphs (i.e.\ of maximum degree three).
The Shannon's bound gives $\gamma(G)\ge\frac{3}{4}$ then, which is tight by $K_3+e$.
Rizzi~\cite{R09} showed that when $G$ is a subcubic multigraph with no cycles of length 3, then $\gamma(G)\ge\frac{13}{15}$, which is tight by the Petersen graph. Kami\'nski and Kowalik~\cite{swat,Bey} extended this result and proved that $\gamma(G)\ge\frac{7}{9}$ when $G$ is a connected subcubic multigraph different from $K_3+e$.

\subsection{Our Combinatorial Results}
\label{sec:comb-res}

In what follows, for a nonnegative integer $c$ by $cK_3$ we denote the graph on three vertices with every pair of vertices connected by c parallel edges (see Fig.~\ref{fig:2k3}), while $cK_3+e$ denotes the graph that can be obtained from $cK_3$ by adding a new edge between one pair of vertices (see Fig.~\ref{fig:2k3+e}).

\tikzstyle{every node} = [circle, draw, fill=black]
\begin{figure}
\begin{center}
\begin{minipage}{0.3\textwidth}
  \centering
  \begin{tikzpicture}[scale=0.7,thick]
    \node at (0,0) (u) {};
    \node at (4,0) (v) {};
    \node at (2,3.46) (w) {};
    \draw (u) edge[out=15, in=165] (v);
    \draw (u) edge[out=-15, in=195] (v);
    \draw (u) edge[out=75, in=225] (w);
    \draw (u) edge[out=45, in=255] (w);
    \draw (v) edge[out=105, in=-45] (w);
    \draw (v) edge[out=135, in=-75] (w);
  \end{tikzpicture}
  
  \caption{\label{fig:2k3}$2K_3$}
  \end{minipage}
  \begin{minipage}{0.3\textwidth}
  \centering
  \begin{tikzpicture}[scale=0.7,thick]
    \node at (0,0) (u) {};
    \node at (4,0) (v) {};
    \node at (2,3.46) (w) {};
    \draw (u) edge[out=15, in=165] (v);
    \draw (u) edge[out=-15, in=195] (v);
    \draw (u) edge (v);
    \draw (u) edge[out=75, in=225] (w);
    \draw (u) edge[out=45, in=255] (w);
    \draw (v) edge[out=105, in=-45] (w);
    \draw (v) edge[out=135, in=-75] (w);
  \end{tikzpicture}
  
  \caption{\label{fig:2k3+e}$2K_3 + e$}
  \end{minipage}
\end{center}
\end{figure}

Our first result is as follows.

\begin{theorem}\label{thm:main_2}
Let $G$ be a multigraph of maximum degree $\Delta \ge 4$. 
Then $G$ has a $\Delta$-edge colorable subgraph with at least ${\Delta}/{(\lfloor\frac{3}{2}\Delta\rfloor-1)}|E|$ edges, unless
$\Delta$ is even and $G$ contains $\frac{\Delta}{2}K_3$ as a subgraph or
$\Delta$ is odd  and $G$ contains $\frac{\Delta-1}{2}K_3+e$ as a subgraph. 
Moreover, the subgraph and its coloring can be found in polynomial time.
\end{theorem}

Theorem~\ref{thm:main_2} essentially means that for every value of maximum degree $\Delta$ there is a {\em single} bottleneck configuration and when we exclude it, we get a better bound. Note that the bounds in Theorem \ref{thm:main_2} are tight. The smallest examples are $\frac{\Delta}{2}K_3$ with one edge removed when $\Delta$ is even and the two 3-vertex multigraphs with $3(\Delta-1)/2$ edges when $\Delta$ is odd ($\frac{\Delta-1}{2}K_3$ is one of them).
It is natural to ask whether these are again the only bottlenecks and how does the next bottleneck look like.
We partially answer this question, at least for large value of $\Delta$, with the following theorem.

\begin{theorem}\label{thm:main_1}
Let $G$ be a multigraph of maximum degree $\Delta$ and let $t$ be an integer such that $\floor{\frac{3\Delta}{2}}\geq t\geq\left(\frac{1}{2}\sqrt{22}-1\right)\Delta\approx 1.34\Delta$. 
Assume that $G$ does not contain a three-vertex subgraph with more than $t$ edges. 
Then $G$ has a $\Delta$-edge-colorable subgraph with at least $\frac{\Delta}{t}|E|$ edges.
Moreover, the subgraph and its coloring can be found in polynomial time.
\end{theorem}

By Vizing's theorem for multigraphs, every multigraph of maximum degree $\Delta$ and maximal edge multiplicity $\mu$ has a $\Delta$-edge colorable subgraph with at least $\tfrac{\Delta}{\Delta+\mu}$ edges. Below we state an immediate corollary from Theorem~\ref{thm:main_1} which improves this bound for $\mu \in [\tfrac{\sqrt{22}-2}{6}\Delta,\Delta/2)$.

\begin{corollary}
Let $G$ be a multigraph of maximum degree $\Delta$ and maximal edge multiplicity $\mu\geq\frac{\Delta}{6}\left(\sqrt{22}-2\right)\approx 0.45\Delta.$ Then $G$ has a $\Delta$-edge-colorable subgraph with at least $\frac{\Delta}{3\mu}|E|$ edges.
Moreover, the subgraph and its coloring can be found in polynomial time.
\end{corollary}

Below we state a (not immediate) corollary from Theorem~\ref{thm:main_2}. It will be useful in applications in approximation algorithms.

\begin{theorem}\label{thm:main}
Let $G$ be a connected multigraph of maximum degree $\Delta\ge 3$. Then $G$ has a $\Delta$-edge-colorable subgraph with at least 
\begin{enumerate}
 \item $\frac{2\Delta}{3\Delta-2}|E|$ edges when $\Delta$ is even and $G\ne \frac{\Delta}{2}K_3$,
 \item $\frac{2\Delta+1}{3\Delta}|E|$ edges when $\Delta$ is odd and $G\ne \frac{\Delta-1}{2}K_3+e$.
\end{enumerate}
Moreover, the subgraph and its coloring can be found in polynomial time.
\end{theorem}


Again, the bounds in Theorem \ref{thm:main} are tight. The smallest examples are $\frac{\Delta}{2}K_3$ with one edge removed when $\Delta$ is even  and a graph consisting of two copies of $\frac{\Delta-1}{2}K_3+e$ with the two vertices of degree $\Delta-1$ joined by an edge when $\Delta$ is odd. 

\subsection{Approximation Algorithms}
One may ask why we study $\gamma_{\Delta}(G)$ and not, say $\gamma_{\Delta+1}(G)$. Our main motivation is that finding large $\Delta$-edge-colorable subgraphs has applications in approximation algorithms for the {\sc Maximum $k$-Edge-Colorable Subgraph} problem (aka Maximum Edge $k$-coloring~\cite{FOW02}). In this problem, we are given a graph $G$ (without any restriction on its maximum degree) and the goal is to compute a maximum $k$-edge colorable subgraph of $G$. It is known to be APX-hard when $k \geq 2$ \cite{CP80, H81, LG83}. The research on approximation algorithms for max $k$-ECS problem was initiated by Feige, Ofek and Wieder~\cite{FOW02}. (In the discussion below we consider only multigraphs, consult~\cite{Bey} for an overview for simple graphs.)

Feige et al.~\cite{FOW02} suggested the following simple strategy.
Begin with finding a maximum $k$-matching $F$ of the input graph, i.e.\ a subgraph of maximum degree $k$ which has maximum number of edges. This can be done in polynomial time (see e.g.~\cite{schrijver}). Since a $k$-ECS is a $k$-matching itself, $F$ has at least as many edges as the maximum $k$-ECS. Hence, if we color $\rho|E(F)|$ edges of $F$ we get a $\rho$-approximation.
If we combine this algorithm with (the constructive version of) Shannon's Theorem, we get $k/\lfloor\frac{3}{2}k\rfloor$-approximation.
By plugging in the Vizing's theorem for multigraphs, we get a $\frac{k}{k+\mu(G)}$-approximation, where $\mu(G)$ is the maximum edge multiplicity.

Feige et al.~\cite{FOW02} show a polynomial-time algorithm which, for a given multigraph and an integer $k$, finds a subgraph $H$ such that $|E(H)|\ge {\rm OPT}$, $\Delta(H)\le k+1$ and $\Gamma(H)\le k+\sqrt{k+1}+2$, where ${\rm OPT}$ is the number of edges in the maximum $k$-edge colorable subgraph of $G$, and $\Gamma(H)$ is the odd density of $H$, defined as $\Gamma(H)=\max_{S\subseteq V(H), |S|\ge 2}\frac{|E(S)|}{\lfloor |S|/2\rfloor}$. The subgraph $H$ can be edge-colored with at most $\max\{\Delta+\sqrt{\Delta/2},\lceil \Gamma(H)\rceil\}\le\lceil k + \sqrt{k+1} + 2 \rceil$ colors in $n^{O(\sqrt{k})}$-time by an algorithm of Chen, Yu and Zang~\cite{chen-jco}. By choosing the $k$ largest color classes as a solution this gives a $k/\lceil k + \sqrt{k+1} + 2 \rceil$-approximation. One can get a slightly worse $k/(k+(1+3/\sqrt{2})\sqrt{k} + o(\sqrt{k}))$-approximation by replacing the algorithm of Chen et al. by an algorithm of Sanders and Steurer~\cite{ss} which takes only $O(nk(n+k))$-time.
Note that in both cases the approximation ratio approaches 1 when $k$ approaches $\infty$.

{
  
\linespread{1.1}
\begin{table}[ht]
       \begin{center}
\begin{tabular}{|c||c|c||c|c|}
                       \hline
                       $k$ & previous ratio & reference & new ratio \\
                       \hline                  \hline
                       2 & $\frac{10}{13}$ & \cite{FOW02} & \\
                       3 & $\frac{7}{9}$ & \cite{swat,Bey} & \\
                       4 & $1-(\tfrac{3}4)^4$ & \cite{FOW02}& $\frac{5}{7}$ \\
                       5 & $\tfrac{5}{7}$ & \cite{shannon,FOW02} &  ${\frac{11}{15}}$ \\
                       $6,8,10,12$ & $\max\{\frac{2k}{3k},\frac{k}{k+\mu}\}$ & \cite{shannon,vizing,FOW02} &  $\max\{\frac{2k+2}{3k+2},\frac{k}{k+\mu}\}$ \\
                       $7,9, 11, 13$ & $\max\{\frac{2k}{3k-1},\frac{k}{k+\mu}\}$ & \cite{shannon,vizing,FOW02} &  $\max\{\frac{2k+1}{3k},\frac{k}{k+\mu}\}$ \\
                       $\ge 14$ & $\max\{\frac{k}{\lceil k+\sqrt{k+1}+2\rceil},\frac{k}{k+\mu}\}$ & \cite{chen-jco,vizing,FOW02} & \\
                       \hline
               \end{tabular}
      \end{center}
       \caption{\label{piekna-tabelka}Approximating Maximum $k$-Edge-Colorable Subgraph in multigraphs}
\end{table}
 
}
The results above work for all values of $k$. However, for small values of $k$ tailor-made algorithms are known, with much better approximation ratios. 
Feige et al.~\cite{FOW02} proposed a $\frac{10}{13}$-approximation algorithm for $k=2$ based on an LP relaxation.
They also analyzed a simple greedy algorithm and showed that it has approximation ratio $1-(1-\frac{1}k)^k$, which is still the best result for the case $k=4$ in multigraphs. For $k=3$ Shannon's bound gives a $3/4$-approximation. However, Kami\'nski and Kowalik~\cite{swat,Bey} showed that $K_3+e$ is the only tight example for the Shannon's bound in subcubic graphs; otherwise $\gamma(G)\le \frac{7}{9}$.
One cannot combine this result directly with the $k$-matching technique, since the $k$-matching may contain components isomorphic to $K_3+e$. However, inspired by a paper of Kosowski~\cite{K09}, Kami\'nski and Kowalik~\cite{swat,Bey} showed a general algorithmic technique which leads to improved approximation factors even if the bound on $\gamma(G)$ does not hold for a few special graphs. Using this technique they get a $\frac{7}{9}$-approximation for $k=3$.  

In this paper we also apply the constructive versions of our combinatorial bounds with the algorithmic technique from~\cite{swat,Bey} and we obtain new approximation algorithms which improve the previously known approximation ratios for $4\leq k\leq 13.$ The current state of art in approximating {\sc Maximum $k$-Edge-Colorable Subgraph} for multigraphs is given in Table~\ref{piekna-tabelka}.

\subsection{Preliminaries} 
Our notation is mostly consistent with the one used in \cite{Bey}, which we recall below.

Let $G=(V,E)$ be a graph. For a vertex $x\in V$ by $N(x)$ we denote the set of neighbors of $x$ and $N[x]=N(x)\cup\{x\}$. For a set of vertices $S$ we denote $N(S)=\bigcup_{x\in S}N(x)\setminus S$ and $N[S]=\bigcup_{x\in S}N[x]$. 
We also denote by $I[S]$ the subgraph whose set of vertices is $N[S]$ and set of edges is the set of edges of $G$ incident with $S$. For a subgraph $H$ of $G$ we denote $N[H]=N[V(H)]$ and $I[H]=I[V(H)]$.


A {\em partial $k$-coloring} of a graph $G=(V,E)$ is a function $\pi : E \rightarrow \{1, \ldots, k\} \cup \{\bot\}$ such that if two edges $e_1, e_2\in E$ are incident then $\pi(e_1)\neq\pi(e_2)$, or $\pi(e_1)=\pi(e_2)=\bot$. From now on by a {\em coloring} of a graph we will mean a partial $\Delta(G)$-coloring. We say that an edge $e$ is {\em uncolored} if $\pi(e)=\bot$; otherwise, we say that $e$ is {\em colored}. For a vertex $v$, ${\pi}(v)$ is the set of colors of edges incident with $v$, i.e.\ $\pi(v)=\{\pi(e)\ :\ e\in I[v]\}\setminus\{\bot\}$, while $\overline{\pi}(v)=\{1, \ldots, k\}\setminus \pi(v)$ is the set of free colors at $v$. 

Let $V_{\bot} = \{v \in V\ :\ \f(v)\ne \emptyset\}$. In what follows, $\bot(G,\pi)=(V_{\bot},\pi^{-1}(\bot))$ is called {\em the graph of free edges}.
Every connected component of the graph $\bot(G,\pi)$ is called a {\em free component}. If a free component has only one vertex, it is called {\em trivial}.

Below we state a few lemmas proved in \cite{Bey} which will be useful in the present paper. Although the lemmas were formulated for simple graphs one can easily check that the proofs apply to multigraphs as well.

\begin{lemma}[\cite{Bey}, Lemma 7]
\label{lem:distinct-free}
Let $(G,\pi)$ be a colored graph that maximizes the number of colored edges.
For any free component $Q$ of $(G,\pi)$ and for every two distinct vertices $v,w\in V(Q)$
\begin{enumerate}[(a)]
 \item $\f(v) \cap \f(w) = \emptyset$,
 \item for every $a\in \f(v)$, $b\in\f(w)$ there is an $(ab, vw)$-path. 
\end{enumerate}
\end{lemma}

For a free component $Q$, by $\f(Q)$ we denote the set of free colors at the vertices of $Q$, i.e.\ $\f(Q)=\bigcup_{v\in V(Q)}\f(v)$.

\begin{corollary}[\cite{Bey}, Lemma 8]
\label{cor:num-uncolored-0}
Let $(G,\pi)$ be a colored graph that maximizes the number of colored edges.
For any free component $Q$ of $(G,\pi)$ we have $|\f(Q)|\ge 2|E(Q)|$. In particular $Q$ has at most $\left\lfloor\frac{\Delta}{2}\right\rfloor$ edges.
\end{corollary}

Let $Q_1,Q_2$ be two distinct free components of $(G,\pi)$. Assume that for some pair of vertices $x\in V(Q_1)$ and $y\in V(Q_2)$, there is an edge $xy\in E$ such that $\pi(xy)\in\f(Q_1)$. Then we say that {\em $Q_1$  sees $Q_2$ with $xy$}, or shortly {\em $Q_1$ sees $Q_2$}.

\begin{lemma}[\cite{Bey}, Lemma 10]
\label{lem:A}
Let $(G,\pi)$ be a colored graph that maximizes the number of colored edges.
If $Q_1,Q_2$ are two distinct free components of $(G,\pi)$ such that $Q_1$ sees $Q_2$ then
$\f(Q_1) \cap \f(Q_2) = \emptyset$.
\end{lemma}

We use the notion of the potential function $\Psi$ introduced in \cite{Bey}:
$$\Psi(G,\pi)=(c,n_{\floor{\Delta/2}},n_{\floor{\Delta/2}-1},\ldots,n_1),$$
where $c$ is the number of colored edges, i.e.\ $c=|\pi^{-1}(\{1,\ldots,\Delta\})|$ and $n_i$ is the number of free components with $i$ edges for every $i=1,\ldots,\floor{\Delta/2}$.

\subsection{Our Approach and Organization of the Paper}
Informally, our plan for proving the main results is to consider a coloring that maximizes the potential $\Psi$ and injectively assign many colored edges to every free component in the coloring. To this end we introduce edges controlled by a component (each of them will be assigned to the component which controls it) and edges influenced by a component (as we will see every edge is influenced by at most two components; if it is influenced by exactly two components, we will assign {\em half} of the edge to each of the components).

In Section~\ref{sec:struct} we develop structural results on colorings maximizing $\Psi$. Informally, these results state that in such a coloring every free component influences/controls many edges. Then, in Subsection~\ref{sec:charge} we prove lower bounds for the number of edges assigned to various types of components, using a convenient formalism of sending charge.

In Section~\ref{sec:collapsing} we develop a method of collapsing subgraphs. We use it to reduce some special graphs, for which maximizing $\Psi$ does not give a sufficiently good result, to graphs better suited to our needs.
The paper concludes with Section~\ref{sec:proof} containing the proofs of the main results.

\section{The structure of colorings maximizing $\Psi$}
\label{sec:struct}

\subsection{Moving free components}

Note that if $P$ and $Q$ are distinct free components of a coloring $(G,\pi)$ then $E(P)\cap E(I[Q])=\emptyset$. 

\begin{definition}
Let $(G,\pi)$ be a colored graph 
and let $P$ be a nontrivial free component of $(G,\pi)$. An {\em elementary move} of $P$ in $\pi$ is a coloring $\pi '$ such that:
\begin{enumerate}
\item $\pi '$ can be obtained from $\pi$ by uncoloring $k$ edges of $I[P]\setminus E(P)$ and coloring $k$ edges of $P$ for some $k\geq 0$,
\item $\pi '|_{I[P]}$ has exactly one nontrivial free component, denote it $P'$.
\end{enumerate}
If the above holds we say that $\pi '$ and $P'$ have been obtained respectively from $\pi$ and $P$ by an elementary move.
Sometimes we will write shortly {\em $\pi'$ is an elementary move of $\pi$}, meaning that $\pi'$ is an elementary move of a free component of $\pi$.
\end{definition}

Note that in particular $\pi$ is the trivial elementary move of any of its components.
Note also that $\pi$ and $\pi '$ have the same number of uncolored edges. Furthermore the component $P$ is either replaced with a component $P'$ or merged with a component $Q$ into $P'\cup Q$. Either way an elementary move does not decrease the potential $\Psi$. Furthermore we have the following:

\begin{remark}\label{ruch_potencjal}
If $\pi$ maximizes the potential $\Psi$ then an elementary move $\pi'$ of a component $P$ cannot cause a merge of nontrivial components and hence $P'$ is a free component of $\pi '$.
\end{remark}

Thus if $\pi$ maximizes the potential $\Psi$ and $\pi '$ is an elementary move of $\pi$ then there is a one-to-one correspondence between free components of $\pi$ and free components of $\pi '$. For a free component $P'$ of $\pi '$ and a free component $P$ of $\pi$ we write that $P'=P(\pi)$ if either $P=P'$ or $\pi '$ is the elementary move of $P$ to $P'$.

\begin{lemma}\label{lem:elem-move-local}
 Let $(G,\pi)$ be a coloring that maximizes the potential $\Psi$ and let $P$ and $Q$ be two distinct nontrivial free components. 
 Let $\pi'$ be an elementary move of $P$ in $\pi$.
 Then, $\pi'|_{E(I[Q])} = \pi|_{E(I[Q])}$.
\end{lemma}

\begin{proof}
 Recall that $\pi'$ is obtained from $\pi$ by uncoloring some edges $E_1\subseteq E(I[P])$, and coloring some edges $E_2\subseteq E(P)$.
 Then $E_2\cap E(I[Q])=\emptyset$ since $P$ and $Q$ are distinct free components.
 Moreover, $E_1\cap E(I[Q])=\emptyset$ for otherwise $\Psi(G,\pi')>\Psi(G,\pi)$.
\end{proof}

\begin{lemma}\label{przemienne_ruchy}
Let $(G,\pi)$ be a coloring that maximizes the potential $\Psi$ and let $P$ and $Q$ be two distinct nontrivial free components. Suppose that $P'$ and $Q'$ can be obtained respectively by elementary moves $\pi_P$ and $\pi_Q$ of $P$ and $Q$. Then:
\begin{enumerate}[$(i)$]
\item $Q$ is a free component of $\pi_P$,
\item $P$ is a free component of $\pi_Q$,
\item $\pi_P|_{E(I[P])\cap E(I[Q])}=\pi_Q|_{E(I[P])\cap E(I[Q])}$,
\item let us define $\pi':E(G)\rightarrow\{1,\ldots,\Delta\}\cup\{\bot\}$ as follows
\begin{eqnarray*}
\pi'|_{E(G)\setminus(E(I[P])\cup E(I[Q]))}&=&\pi|_{E(G)\setminus(E(I[P])\cup E(I[Q]))}, \\
\pi'|_{E(I[P])}&=&\pi_P|_{E(I[P])}, \\ 
\pi'|_{E(I[Q])}&=&\pi_Q|_{E(I[Q])}.
\end{eqnarray*}
Then, $\pi'$ is an elementary move of $Q$ in $\pi_P$ such that $Q(\pi')=Q'$ and
      $\pi'$ is an elementary move of $P$ in $\pi_Q$ such that $P(\pi')=P'$.
\end{enumerate}
\end{lemma}
\begin{proof}
$(i)$ and $(ii)$ follow directly from Lemma~\ref{lem:elem-move-local}. 
For $(iii)$, note that if $e\in E(I[P])\cap E(I[Q])$ then $e\not\in E(Q)$, since $P$ and $Q$ are distinct free components.
Hence, by Lemma~\ref{lem:elem-move-local}, $\pi_P(e)=\pi(e)$. By symmetry, also $\pi_Q(e)=\pi(e)$, which proves the claim.
Note that $(iii)$ implies that $\pi'$ is well-defined in $(iv)$.
If we compare $\pi'$ to $\pi$ we see that only edges in $E(P)$ and $E(Q)$ get new colors. Since the former are incident only to edges in $I(P)$ and the latter only to edges in $I(Q)$, and $\pi_P$ and $\pi_Q$ are proper partial edge-colorings (i.e.\ incident edges get different colors), so $\pi'$ is also a proper partial edge-coloring. The rest of the claim follows from Lemma~\ref{lem:elem-move-local}.
\end{proof}

\begin{lemma}\label{lem:ruch_id}
Let $(G,\pi)$ be a coloring that maximizes the potential $\Psi$. Let $\pi_0, \pi_1, \ldots, \pi_k$ be a sequence of colorings such that $\pi_0=\pi_k=\pi$ and for $i=1,\ldots,k$ coloring $\pi_i$ is an elementary move of a free component in $\pi_{i-1}$. Let $P_0$ be a free component of $\pi_0$ and let $P_i=P_{i-1}(\pi_i)$. Then $P_k=P_0$.
\end{lemma}
\begin{proof}
Suppose that for some index $j$ it holds that $\pi_{j+1}$ is an elementary move of $P_j$ in $\pi_j$ and $\pi_j$ is an elementary move of some free component $R$ of $\pi_{j-1}$ distinct from $P_{j-1}$. Then $P_j=P_{j-1}$.
Define a coloring $\pi_j^*$ as follows: $\pi^*_j|_{E(I[P_{j-1}])}=\pi_{j+1}|_{E(I[P_{j-1}])}$ and $\pi^*_j|_{E(G)\setminus E(I[P_{j-1}])}=\pi_{j-1}|_{E(G)\setminus E(I[P_{j-1}])}$.
Then by Lemma~\ref{lem:elem-move-local}, $\pi^*_j$ is an elementary move of $P_{j-1}$ in $\pi_{j-1}$ and $\pi^*_j(P_{j-1}) = P_{j+1}$.
Moreover, by Lemma~\ref{lem:elem-move-local}, $\pi_{j+1}|_{E(G)\setminus(E(I[R])\cup E(I[P_{j-1}]))}=\pi_{j-1}|_{E(G)\setminus(E(I[R])\cup E(I[P_{j-1}]))}$, $\pi_{j+1}|_{E(I[R])}=\pi_j|_{E(I[R])}$, and $\pi_{j+1}|_{E(I[P_{j-1}])}=\pi_j^*|_{E(I[P_{j-1}])}$.
Hence, by Lemma~\ref{przemienne_ruchy}$(iv)$, $\pi_{j+1}$ is an elementary move of $R$ in $\pi_j^*$.  Thus we may replace $\pi_j$ with $\pi_j^*$ in the sequence $\pi_0, \pi_1, \ldots, \pi_k$ and redefine $P_j$ accordingly. Note that by replacing $\pi_j$ with $\pi_j^*$ we have decreased the sum of indexes $j$ such that $\pi_{j+1}$ is an elementary move of $P_j$ in $\pi_j$. We continue this process as long as possible and obtain a sequence of colorings and an index $0\leq i_0\leq k$ such that for every $1\leq i\leq i_0$ $\pi_i$ is an elementary move of $P_{i-1}$ in $\pi_{i-1}$ and for every $i_0< i\leq k$ $\pi_i$ is an elementary move of some free component of $\pi_{i-1}$ distinct from $P_{i-1}$. Note that $\pi_
0$ and $\pi_k$ were never changed.

Suppose that $P_0$ and $P_k$ are two distinct free components of $\pi$. 
Let $Q_0=P_k$ and $Q_i=Q_{i-1}(\pi_i)$. By induction we obtain that $P_i$ and $Q_i$ are distinct free components of $\pi_i$. Note that for $i\leq i_0$ we have $Q_i=Q_{i-1}$ and for $i>i_0$ we have $P_i=P_{i-1}$. Thus $Q_{i_0}=Q_0=P_k$ and $P_{i_0}=P_k$, a contradiction.
\end{proof}

We say that a coloring $\pi '$ is a {\em move} of $\pi$ if $\pi$ maximizes the potential $\Psi$ and there is a sequence of colorings $\pi_0=\pi, \pi_1, \ldots, \pi_k=\pi '$ such that $\pi_i$ is an elementary move of a free component in $\pi_{i-1}$. 
We say that a free component $P'$ of $\pi '$ is obtained from a free component $P$ in $\pi$ if for every $i=0,\ldots,k$ there is a free component $P_i$ of $\pi_i$ such that $P_0=P$, $P'=P_k$ and for $i>0$ we have $P_i=P_{i-1}(\pi_i)$. We denote the free component of a coloring $\pi '$ obtained from a free component $P$ by $P(\pi ')$. This notation extends the notation introduced earlier for elementary moves. Note that $P(\pi ')$ does not depend on the sequence $\pi_1,\ldots,\pi_{k-1}$ since given a different sequence $\pi_1',\ldots,\pi_{t-1}'$ we may apply Lemma~\ref{lem:ruch_id} to the sequence $\pi ',\pi_{k-1},\ldots,\pi_1,\pi,\pi_1',\ldots,\pi_{t-1}',\pi '$.

\begin{theorem}\label{niezalezne_ruchy}
Let $(G,\pi)$ be a coloring that maximizes the potential $\Psi$. Then:
\begin{enumerate}[$(i)$]
\item if $\pi '$ is a move of $\pi$ then $\pi$ is a move of $\pi '$;
\item  Let $P$ be a nontrivial free component of $\pi$ and let $\pi '$ be a move of $\pi$. 
Then there is a sequence of elementary moves $\{\pi_i\}_{i=1,\ldots,k}$, such that for every $i=1,\ldots,k$, $\pi_i$ is a move of a component $P_{i-1}$ in $\pi_{i-1}$, where $P_0=P$, $\pi_0=\pi$, for every $i=1,\ldots,k$ we have $P_i=P(\pi_i)$, $P_k=P(\pi')$ and $\pi_k|_{E(I[P_k])}=\pi '|_{E(I[P_k])}$;
\item if $P_1,\ldots,P_k$ are distinct nontrivial free components of $\pi$ and $\pi_1,\ldots,\pi_k$ are moves of $\pi$ then there is a move $\pi '$ of $\pi$ such that $P_i(\pi ')=P_i(\pi_i)$ and $\pi '|_{E(I[P_i(\pi ')])}=\pi_i|_{E(I[P_i(\pi_i)])}$ for $i=1,\ldots,k$.
\end{enumerate}
\end{theorem}
\begin{proof}~
\begin{enumerate}[$(i)$]
\item By Remark~\ref{ruch_potencjal} if $\pi '$ is an elementary move of a free component $P$ in $\pi$ then $\pi$ is an elementary move of $P(\pi ')$ in $\pi '$. Since a move is a sequence of elementary moves the claim follows.

\ignore{\item Let $\pi_1,\ldots,\pi_k$ be a shortest sequence of elementary moves giving a move such that $P(\pi_k)=P(\pi ')$ and $\pi_k|_{E(I[P(\pi_k)])}=\pi '|_{E(I[P(\pi_k)])}$. Let $P_0=P$ and $P_i=P_{i-1}(\pi_i)$. We claim that $\pi_i$ are elementary moves of $P_{i-1}$. Suppose this is not true and let $\pi_j$ be the elementary move of $Q_{j-1}\neq P_{j-1}$ with the highest index. By using Lemma \ref{przemienne_ruchy} multiple times we see that $\pi_k$ can be obtained from $\pi_{i-1}$ by first moving $P_j,\ldots, P_{k-1}$ and then $Q_{j-1}$. However reversing the moving of $Q_{j-1}$ does neither change $P_k$ nor the colors of $E(I[P_k])$, thus we obtain a sequence of $k-1$ elementary moves giving a move from $P$ to $P'$ --- a contradiction.}

\item Let $\pi_0,\ldots,\pi_t$ be any sequence of elementary moves such that $\pi_0=\pi$ and $\pi_t=\pi '$. 
We use the process described in the proof of Lemma~\ref{lem:ruch_id} to redefine the colorings $\pi_i$ for $0<i<t$ and obtain an index $k$ such that for every $1\leq i\leq k$ $\pi_i$ is an elementary move of $P_{i-1}$ in $\pi_{i-1}$ and for every $k< i\leq t$ $\pi_i$ is an elementary move of some free component of $\pi_{i-1}$ distinct from $P_{i-1}$. 
Note that for $k< i\leq t$ we have $P_i=P_{i-1}$ and by Lemma~\ref{lem:elem-move-local} it holds that $\pi_i|_{E(I[P_i])}=\pi_{i-1}|_{E(I[P_{i-1}])}$, hence $\pi_1,\ldots,\pi_k$ is the desired sequence.

\item We will show the claim for $k=2$. The proof for more components is a trivial generalization of this one but involves a multitude of indices. Using $(ii)$ we get two sequences of elementary moves $\pi_{1,0},\ldots,\pi_{t_1+1,0}$ and $\pi_{0,1},\ldots,\pi_{0,t_2+1}$ of free components $P_1^0,\ldots,P_1^{t_1}$ and $P_2^0,\ldots,P_2^{t_2}$ giving respectively a move from $P_1=P_1^0$ to $P_1(\pi_1)$ and a move from $P_2=P_2^0$ to $P_2(\pi_2)$ such that all $P_1^{j_1}$ are free components obtained from $P_1$ and all $P_1^{j_1}$ are free components obtained from $P_1$. 
We claim that for every $j_1=1,\ldots,t_1+1$ and $j_2=1,\ldots,t_2+1$ there exists a move $\pi_{j_1,j_2}$ of $\pi$ such that $P_1(\pi_{j_1,j_2})=P_1(\pi_{j_1,0})$ and $P_2(\pi_{j_1,j_2})=P_2(\pi_{0,j_2})$. 
Furthermore we require $\pi_{j_1,j_2}$ to simultaneously be an elementary move of $P_1(\pi_{j_1-1,j_2})$ in $\pi_{j_1-1,j_2}$ and $P_2(\pi_{j_1,j_2-1})$ in $\pi_{j_1,j_2-1}$. By Lemma~\ref{przemienne_ruchy}$(iv)$ the existence of $\pi_{j_1,j_2}$ follows from the 
existence of 
$\pi_{j_1-1,j_2}$ and $\pi_{j_1,j_2-1}$ thus the claim follows by induction on $j_1+j_2$. The fact that $\pi '|_{E(I[P_i(\pi ')])}=\pi_i|_{E(I[P_i(\pi_i)])}$ follows from the explicit definition of coloring $\pi '$ in Lemma \ref{przemienne_ruchy}.
\end{enumerate}
\end{proof}

\subsection{Controlling vertices}
Let $P$ be a nontrivial free component of $(G,\pi)$.
By $\mathcal{M}(P)$ we denote the set of moves of $\pi$ that can be obtained by moving only $P$ (and the components obtained from $P$). 
We say that vertex $v$ is {\em controlled} by $P$ when $v\in V(P(\pi'))$ for some $\pi'\in\mathcal{M}(P)$.
By $\con(P)$ we denote the set of vertices controlled by $P$.

\begin{lemma}\label{lem:con-disjoint}
Let $(G,\pi)$ be a coloring that maximizes the potential $\Psi$ and let $Q_1$ and $Q_2$ be two distinct nontrivial free components of $\pi$. Then $\con(Q_1)\cap\con(Q_2)=\emptyset$.
\end{lemma}

\begin{proof}
Assume there is a vertex $v\in \con(Q_1)\cap\con(Q_2)$.
Then, for each $i=1,2$, there is a move $\pi_i\in\mathcal{M}(Q_i)$ such that $v\in V(Q_i(\pi_i))$. 
By Theorem \ref{niezalezne_ruchy} there is a coloring $\pi'$ with free components $Q_i(\pi_i)$, for $i=1,2$. 
Hence $v$ belongs to two distinct components in $(G,\pi')$, a contradiction.
\end{proof}

\begin{lemma}\label{lem:con-invariant}
Let $(G,\pi)$ be a coloring that maximizes the potential $\Psi$ and let $Q$ be a nontrivial free component of $\pi$. Let $\pi '$ be a move of $\pi$ and let $Q'=Q(\pi ')$. Then $\con(Q)=\con(Q')$.
\end{lemma}

\begin{proof}
Pick any $v\in\con(Q)$. Let $\pi_1\in\mathcal{M}(Q)$ be such that $v\in V(Q(\pi_1))$. Note that $\pi_1$ is a move of $\pi '$ and $Q'(\pi_1)=Q(\pi_1)$. By Theorem~\ref{niezalezne_ruchy}(ii) there is some $\pi_2\in\mathcal{M}(Q')$ such that $Q'(\pi_1)=Q'(\pi_2)$. Thus $v\in\con(Q')$ and $\con(Q)\subset\con(Q')$. By symmetry $\con(Q')\subset\con(Q)$ and the claim follows.
\end{proof}

In what follows we define a subset of $\con(Q)$, which will be particularly useful.

\begin{definition}\label{def_WQ}
For a free component $Q$ of a coloring $(G,\pi)$ we define a set of vertices $\con_1(Q)=\{v\in V(G)\ :\ v\in V(Q(\pi '))$ for some elementary move $\pi'$ of $Q$ such that $|E(Q(\pi '))\setminus E(Q)|\leq 1\}$.
\end{definition}

\begin{lemma}
\label{wp1}
\label{lem:disjoint-free-W}
Let $(G,\pi)$ be a coloring that maximizes the potential $\Psi$ and let $Q$ be a free component with $|E(Q)|\geq 2$.
If $u,v\in \con_1(Q)$ and $u\neq v$ then $\f(u)\cap\f(v)=\emptyset$.
\end{lemma}

\begin{proof}
Assume that we have two distinct $u,v\in \con_1(Q)$ such that there exists a color
$a\in\f(u)\cap\f(v)$.
By Lemma \ref{lem:distinct-free} the vertices $u$ and $v$ cannot both belong
to $V(Q)$. By symmetry we can assume $u\not\in V(Q)$.

Since $u\in \con_1(Q)$, one can uncolor an edge $ux$, $x\in V(Q)$, and color an edge $e\in E(Q)$, obtaining a proper coloring $\pi'$.
Note that $e$ is incident with $x$ for otherwise we can color $e$ without uncoloring $ux$, and increase the potential $\Psi$.
Hence $e=xy$ for some $y\in V(Q)$ and $\pi(ux)\in \f(y)$. 
It follows that $v\not\in V(Q)$ for otherwise $Q$ sees the (possibly trivial) component containing $u$, a contradiction with Lemma~\ref{lem:A}.
Let $vp$ be the edge such that one can uncolor $vp$ and color an edge of $Q$ with color $\pi(vp)$, which exists because $v\in \con_1(Q)$.
By Corollary~\ref{cor:num-uncolored-0}, $|\f(Q)|\ge 4$, so there is a color $b\in\f(Q)\setminus\{\pi(ux),\pi(vp)\}$. Let $z$ be the vertex of $Q$ such that $b\in\f(z)$. Note that $a\not\in\f(z)$ for otherwise $Q$ sees the (possibly trivial) component containing $v$, a contradiction with Lemma~\ref{lem:A}.
Consider a maximal path $P$ which starts at $z$ and has edges colored in $a$ and $b$ alternately. Swap the colors $a$ and $b$ on the path. As a result, $a$ becomes free in $Q$. Also, $P$ touches at most one of the vertices $u$, $v$, so $a$ is still free in at least one of them, by symmetry say in $u$. Since $b\ne\pi(ux)$, the vertex $u$ is still in $\con_1(Q)$. Hence we arrive at the first case again.
\end{proof}

\subsection{Sending charges}
\label{sec:charge}

In this section we consider a connected colored graph $(G,\pi)$ which maximizes the potential $\Psi$.

We put one unit of charge on each colored edge of $G$. 
Every edge divides its charge equally between the nontrivial components that control its endpoints (by Lemma~\ref{lem:con-disjoint} there are at most two such components).
For every nontrivial free component $P$, let $\ch(P)$ denote the amount of charge sent to $P$. Then the number of colored edges in $G$ is at least $\sum_P \ch(P)\geq|E|\min_P\frac{\ch(P)}{\ch(P)+|E(P)|}$, where the summation is over all nontrivial free components in $(G,\pi)$. In what follows we give lower bounds for $\ch(P)$ for various types of free components.

\begin{observation}\label{niezmienne_ladunki}\label{obs:niezmienne_ladunki}
Let $P$ be a free component of $(G,\pi)$ and let $\pi '$ be a move of $\pi$.
Then every edge of $G$ sends the same amount of charge to $P(\pi ')$ in the coloring $\pi '$ as it does to $P$ in $\pi$.
\end{observation}

\begin{proof}
By Lemma~\ref{lem:con-invariant} for an arbitrary free component $Q$ of $\pi$ we have $\con(Q(\pi'))=\con(Q)$. Thus $P$ controls an endpoint of an edge iff $P(\pi ')$ controls this endpoint. Furthermore $P$ is the only free component of $\pi$ controlling an endpoint of an edge iff $P(\pi ')$ is the only free component of $\pi '$ controlling an endpoint of that edge.
\end{proof}

For a subset $S\subseteq E$ by $\deg_Sv$ we denote the degree of $v$ in the graph $(V,S)$ and $V(S)$ denotes the set of endpoints of all edges in $S$.
The following observation follows immediately from Lemma~\ref{lem:con-disjoint} (we state explicitly also the weaker bound $\frac{1}{2}|V(S)\cap \con(P)|$ because it will be sufficient in many cases).

\begin{observation}
\label{obs:1/2}
Let $P$ be a nontrivial free component and let $S$ be an arbitrary set of colored edges.
Then edges in $S$ send the charge of at least $\tfrac{1}{2}\sum_{v\in\con(P)}\deg_Sv\ge\frac{1}{2}|V(S)\cap \con(P)|$ to $P$.
\end{observation}

Let $Q$ be a free component in a colored graph $(G,\pi_1)$. If a colored edge $e$ is incident with $Q$, $|\overline{\pi_1}(Q)| \ge \Delta(G) - 1$,  and $\pi_1(e) \in \overline{\pi_1}(Q)$, then we say that edge $e$ is {\em dominated} by $Q$ in the coloring $\pi_1$. 

\begin{lemma}
 \label{lem:full-components}
Let $P$ be a nontrivial free component in $(G,\pi)$ and let $\pi_1$ be a move of $\pi$. Let $P'=P(\pi_1)$.
Every edge dominated by $P'$ in $\pi_1$ sends its whole charge to $P$.
\end{lemma}

\begin{proof}
Let $e=xy$ and $x\in V(P')$. (Then $x\in\con(P)$.) By Lemma~\ref{lem:con-disjoint} it suffices to show that there is no nontrivial component $Q$ such that $y\in\con(Q)$.
Assume for a contradiction that such a component exists.
Since $y\in \con(Q)$, there is a move $\pi_2\in\mathcal{M}(Q)$ such that $y\in V(Q(\pi_2))$.
By Theorem \ref{niezalezne_ruchy}(iii) there is a move $\pi_3$ of $\pi$ such that $P(\pi_3)=P'$ and $Q(\pi_3)=Q(\pi_2)$ and $\pi_3|_{E(I[P(\pi_3)])}=\pi_1|_{E(I[P'])}$. In particular $\pi_3(e)=\pi_1(e)$, so $P(\pi_3)$ sees $Q(\pi_3)$ through $e$ in the colored graph $(G,\pi_3)$ and by Lemma~\ref{lem:A} they have disjoint sets of free colors.
However, since $Q$ is nontrivial, by Corollary~\ref{cor:num-uncolored-0} we have $|\f(Q)|\ge 2$, a contradiction.
\end{proof}

\begin{lemma}\label{P1}
Every free component $P$ of $(G,\pi)$ with $|E(P)|=1$ receives at least $\Delta(G)$ of charge.
\end{lemma}
\begin{proof}
We will show that for each color $a\in\{1,\ldots,\Delta\}$ the component $P$ receives at least one unit of charge from edges colored with $a$. 
Suppose that $a\notin\overline{\pi}(P)$. Then each of the two vertices of $P$ is incident with an edge colored with $a$, hence these edges send at least 1 by Observation~\ref{obs:1/2}.
Now suppose that $a\in\overline{\pi}(P)$. 
Denote the vertices of $P$ by $x$ and $y$ so that $a\in\overline{\pi}(x)$.
By Lemma~\ref{lem:distinct-free}, $a\in\pi(y)$;  let $z$ be the neighbor of $y$ such that $\pi(yz)=a$. 
Uncoloring $yz$ and coloring $xy$ is an elementary move so $z\in\con(P)$ and hence $zy$ sends $1$ to $P$ by Observation~\ref{obs:1/2}.
\end{proof}

\begin{lemma}\label{lem:szerokie_komponenty}
Let $P$ be a free component of $(G,\pi)$ and let $U\subseteq\con(P)$ be a set of vertices such that $\f(v)\cap\f(w)=\emptyset$ for every two distinct vertices $v,w\in U$. Then $\ch(P)\geq (|U|-1)|E(P)|$.
\end{lemma}
\begin{proof}
Let $S_a$ denote the set of edges incident to $U$ and colored with color $a$. 
Since $a$ is free in at most one vertex of $U$, we have $|V(S_a)\cap U|\geq |U|-1$. Since $U\subseteq \con(P)$, by Observation~\ref{obs:1/2} we infer that $S_a$ sends at least $\frac{|U|-1}{2}$ of charge to $P$. Taking the sum over all colors we obtain that $P$ receives at least $\frac{\Delta(|U|-1)}{2}$ of charge. Moreover by Lemma \ref{cor:num-uncolored-0} we have $|E(P)|\leq\frac{\Delta}{2}$, hence the claim.
\end{proof}

In what follows, a {\em single edge} is an edge of multiplicity 1.

\begin{lemma}\label{single_delta45}
Let $P$ be a free component of $(G,\pi)$ isomorphic to the $2$-path. Assume that $|\overline{\pi}(P)|\ge\Delta(G)-1$. 
If both edges of $P$ are single edges in $G$, then $\ch(P)\ge 2\Delta(G)$.
\end{lemma}
\begin{proof}
By Lemma~\ref{lem:disjoint-free-W} and Lemma~\ref{lem:szerokie_komponenty}, if $|\con_1(P)|\ge 5$, we are done. In what follows we assume $|\con_1(P)|\leq 4$. 
Let $S_a$ denote the set of edges incident to $\con_1(P)$ and colored with color $a$.
We will show that for every color $a$ the edges of $S_a$ send at least 2 to $P$; then the claim clearly follows.

Let $P=xyz$ and let $a\in\f(x)$. Then, by Lemma~\ref{lem:distinct-free} we have $a\in\pi(y)$, so there is an $yv$ edge such that $\pi(yv)=a$. Since $P$ has single edges we have $v\notin V(P)$. Furthermore we may move $P$ by uncoloring $yv$ and coloring $xy$, thus $v\in \con_1(P)$ and hence $\con_1(P)=\{x,y,z,v\}$.

Consider any edge $e$ colored with $c\in \f(P)\cup\f(v)$ and incident to $\con_1(P)$. If $e$ is incident to $V(P)$ and $c\in\f(P)$ then $e$ sends 1 to $P$ by Lemma~\ref{lem:full-components}. Assume $e$ is incident to $v$ but not to $V(P)$. Then $c\in\f(P)$. Color $c$ is free in at most one of $x,z$; by symmetry we can assume $c\in\pi(x)$. Pick any color $c_x\in\f(x)$. By Lemma~\ref{lem:distinct-free}, $c_x\in\pi(y)$ and hence there is an edge $v'y$ colored with $c_x$. Consider the elementary move $\pi'$ obtained by uncoloring $v'y$ and coloring $xy$ with $c_x$. We obtain $v'\in \con_1(P)$. Since $P$ has single edges it follows that $v=v'$. Furthermore $c\in \overline{\pi'}(P(\pi'))$ and hence $e$ is dominated by $P(\pi')$, so $e$ sends 1 to $P$ by Lemma~\ref{lem:full-components}.
Finally, assume $c\in\f(v)$. Then $e$ is incident with $V(P)$, if both endpoints of $e$ are in $V(P)$ then $e$ sends 1 to $P$. Assume that only one endpoint of $e$ is in $V(P)$; by symmetry we can assume $e$ is not incident with $x$. Then, again $e$ is dominated by $P(\pi')$, so $e$ sends 1 to $P$ by Lemma~\ref{lem:full-components}. 
Hence, for every $c\in\f(P)\cup\f(v)$ every edge in $S_c$ sends 1 to $P$. Moreover, for $c\in\f(P)\cup\f(v)$ we have $|S_c|\ge (|\con_1(P)|-1)/2 = 3/2$, so $|S_c|\ge 2$ and the edges in $S_c$ send $2$ to $P$, as required.

Finally consider $c\not\in\f(P)\cup\f(v)$. Since every vertex in $\con_1(P)$ is incident with an edge in $S_c$, by  Observation~\ref{obs:1/2} we infer that $S_c$ sends at least $\frac{1}{2}|\con_1(P)|=2$ of charge to $P$, as required.
\end{proof}

\begin{lemma}\label{multiple_delta}
Let $P$ be a free component of $(G,\pi)$ consisting of a $k$-fold $xy$ edge for $k\geq 2$. Assume that $|\overline{\pi}(P)|=\Delta(G)$. Then either $\ch(P)\geq 2\Delta(G)$ or $G$ has exactly three vertices.
\end{lemma}
\begin{proof}
First assume $|N(P)|=1$, i.e.\ $N(P)=\{z\}$, for some vertex $z$.
Since $\f(x)\cup\f(y)=\{1,\ldots,\Delta\}$ and $\f(x)\cap\f(y)=\emptyset$ by Lemma~\ref{lem:distinct-free}, so $\pi(x)=\f(y)$ and $\pi(y)=\f(x)$. It follows that the number of edges between $z$ and $P$ is $|\pi(x)|+|\pi(y)|=\Delta$. This means that $z$ has no neighbors outside $P$ and $G$ has exactly three vertices $x$, $y$ and $z$.

Now assume that $z_1$ and $z_2$ are two distinct vertices in $N(P)$. 
For every $i=1,2$ we can uncolor any edge between $z_i$ and $P$ and color one of the edges $xy$ obtaining an elementary move $\pi_i$. 
Thus $z_1,z_2\in \con_1(P)$ and by Lemma~\ref{lem:disjoint-free-W} we have $\overline{\pi}(z_1)=\overline{\pi}(z_2)=\emptyset$. 
Every colored edge incident to $x$ or $y$ is dominated by $P$ and it sends $1$ to $P$ by Lemma~\ref{lem:full-components}. Similarly, for $i=1,2$ every colored edge incident to $z_i$ is dominated by $P(\pi_i)$ and it sends $1$ to $P$ by Lemma~\ref{lem:full-components}. So each colored edge in $I[\{x,y,z_1,z_2\}]$ sends its charge to $P$. Notice that for each color there are at least two distinct edges in $I[\{x,y,z_1,z_2\}]$ colored with this color: one incident with $P$ and one incident with the $z_i$ not connected to $P$ by the first edge. Thus $P$ receives at least two units of charge from edges of a given color and hence the claim.
\end{proof}

\begin{lemma}\label{multiple_delta-1}
Let $P$ be a free component of $(G,\pi)$ consisting of a $k$-fold $xy$ edge for $k\geq 2$. Assume that $|\overline{\pi}(P)|=\Delta(G)-1$. Then either $\ch(P)\geq 2\Delta(G)$ or there is an induced subgraph $H$ of $G$ containing $P$, containing exactly three vertices and connected with $G\setminus H$ by at most three edges which are all colored with the color not in $\overline{\pi}(P)$.
\end{lemma}
\begin{proof}
Denote the color not in $\f(P)$ by $a$.
Let $N_{\neq a}(P)$ denote the vertices from $N(P)$ connected with $P$ by an edge of color different from $a$. 

First assume $|N_{\neq a}(P)|=1$, i.e.\ $N_{\neq a}(P)=\{z\}$ for some vertex $z$. 
Since $\overline{\pi}(P)=\{1,\ldots,\Delta\}\setminus\{a\}$, by Lemma~\ref{lem:distinct-free} for every color $c\ne a$ there is an edge between $z$ and $P$. Then the induced subgraph on vertices $x,y,z$ is joined with the rest of the graph by at most three edges colored with $a$. 

Now assume that $z_1$ and $z_2$ are two distinct vertices in $N_{\neq a}(P)$. 
For every $i=1,2$ we can uncolor any edge between $z_i$ and $P$ and color one of the edges $xy$ obtaining an elementary move $\pi_i$. 
Thus $z_1,z_2\in \con_1(P)$ and by Lemma~\ref{lem:disjoint-free-W} we have $\overline{\pi}(z_1),\overline{\pi}(z_1)\subseteq\{a\}$. 
Hence for each color different from $a$ there are two distinct edges in $I[\{z_1,z_2\}]$ colored with this color (one connects $P$ and $z_i$ for some $i=1,2$, and the other is incident with $z_{3-i}$ and not with $P$) and each of them is dominated by $P$, $P(\pi_1)$ or $P(\pi_2)$, so it sends 1 unit of charge to $P$ by Lemma~\ref{lem:full-components}. It remains to show that $P$ receives at least two units of charge from edges colored with $a$. Denote the set of these edges by $S_a$.

If $\overline{\pi}(z_1)=\overline{\pi}(z_2)=\emptyset$ then $\{x,y,z_1,z_2\}\subseteq V(S_a)$, so by Observation~\ref{obs:1/2}, the edges in $S_a$ send 2 to $P$, as required. Since $z_1,z_2\in \con_1(P)$, by Lemma~\ref{lem:disjoint-free-W} it cannot occur that $\overline{\pi}(z_1)=\overline{\pi}(z_2)=\{a\}$. By symmetry we can assume that $\overline{\pi}(z_1)=\{a\}$ and $\overline{\pi}(z_2)=\emptyset$. 
Both edges colored with $a$ and incident with $P$ are dominated by $P(\pi_1)$ and send 1 unit of charge to~$P$ by Lemma~\ref{lem:full-components}. We are done, unless this is a single $xy$ edge colored with $a$. Assume this is the case, we will complete the proof by showing that the $z_2w$ edge colored with $a$ sends 1 to $P$. 
Notice that we can move $P(\pi_1)$ by uncoloring the $xy$ edge colored with $a$ and coloring the uncolored $xz_1$ or $yz_1$ edge and then further by uncoloring an $xz_2$ or $yz_2$ edge and coloring an $xy$ edge. This way we obtain a coloring $\pi_3$ with the property that $a\in\overline{\pi_3}(P(\pi_3))$, the $z_2w$ edge remains colored with $a$ and 
$z_2$ is a vertex of $P(\pi_3)$. Hence $z_2w$ is dominated by $P(\pi_3)$ and sends 1 to $P$, as required.
\end{proof}

By Lemma~\ref{lem:distinct-free} and Lemma~\ref{lem:szerokie_komponenty}, we immediately obtain the following:

\begin{corollary}\label{cor:4_5for5+}
Let $Q$ be a free component of $(G,\pi)$ such that $|V(Q)|\geq 5$. Then $\ch(Q)\geq 4|E(Q)|$.
\end{corollary}

\begin{lemma}\label{lem:4_5for4}
Let $Q$ be a free component of $(G,\pi)$ such that $|V(Q)|=4$. Then $\ch(Q)\geq 4|E(Q)|$.
\end{lemma}
\begin{proof}
Notice that $Q$ has at least three edges.

\mycase{1.} $|E(Q)|=3$.
For $a\in\{1,\ldots,\Delta\}$, let $S_a$ denote the set of edges incident to $V(Q)$ and colored with color $a$. 
By Lemma~\ref{lem:distinct-free}, we have $|V(S_a)\cap V(Q)|\geq |V(Q)|-1$ for every color $a$.
Hence, by Observation \ref{obs:1/2} edges in $S_a$ send at least $(|V(Q)|-1)/2=\tfrac{3}2$ units of charge to $Q$, so $\ch(Q)\geq \frac{3}{2}\Delta$. If $\Delta\ge 8$, we have $\ch(Q)\ge 12=4|E(Q)|$ and we are done, so in what follows we assume $\Delta \le 7$.
On the other hand, by Corollary~\ref{cor:num-uncolored-0} we infer that $\Delta\ge |\f(Q)|\ge 2|E(Q)| = 6$.
Hence, $\Delta\in\{6,7\}$ and $|\f(Q)|\ge \Delta -1$.
We will show that for every color $a$, edges in $S_a$ send 2 units of charge to $Q$; then $\ch(Q)\geq 2\Delta(G)\geq 12=4|E(Q)|$, as required.
First assume $a\in\f(Q)$. Then every edge in $S_a$ is dominated by $Q$ and by Lemma~\ref{lem:full-components} sends 1 to $Q$. It follows that $S_a$ sends to $Q$ at least $|S_a|\ge\ceil{(|V(Q)|-1)/{2}}=2$ units of charge. 
Finally, if $a\not\in\f(Q)$ then $|V(S_a)\cap V(Q)|=4$, so by Observation~\ref{obs:1/2} edges in $S_a$ send at least $\frac{4}{2}=2$ units of charge to $Q$, as required.

\mycase{2.} $|E(Q)|\ge 4$.
Note that $Q$ contains a cycle (possibly of length 2, i.e.\ a multiple edge). Since $|V(Q)|=4$, all cycles in $Q$ have length at most 4.
Our strategy is to show that $\con(Q)$ contains a subset $U$ of cardinality at least 5 such that every pair of vertices in $U$ has distinct free colors. 
Then the claim follows from Lemma~\ref{lem:szerokie_komponenty}. 

\mycase{2.1.} $Q$ contains a cycle $vpqv$ of length $3$ or a cycle $uvpqu$ of length 4. 
In the prior case let $u$ be the fourth vertex in $V(Q)$ and assume w.l.o.g.\ $uv\in E(Q)$.
Let $a\in\f(u)$.
By Lemma~\ref{lem:distinct-free} each of the vertices $v$, $p$, $q$ is incident  with an edge colored with $a$. 
At least one of these edges has the other endpoint $x$ outside $Q$. By symmetry there are two cases: $(i)$ there is an edge $vx$ colored with $a$ or $(ii)$ $vq$ is colored with $a$ and there is an edge $px$ colored with $a$.
In both cases $x\in\con(Q)$: in case $(i)$ we can move $Q$ by uncoloring $vx$ and coloring $vu$ with $a$, and in case $(ii)$ we can move $Q$ by uncoloring $vq$ and $px$ and coloring $uv$ and $pq$ with $a$. 
Moreover, $\f(x)\cap\f(Q)=\emptyset$, for otherwise $Q$ sees the trivial free component $\{x\}$, a contradiction with Lemma~\ref{lem:A}.
Vertices of $Q$ have disjoint free colors by Lemma~\ref{lem:distinct-free}, so we can put $U=\{u,v,p,q,x\}$.

\mycase{2.2.} $Q$ does not contain cycles, but it contains a multiple edge. 
Then, after ignoring multiplicity of edges, $Q$ is a tree $T_Q$. 
First assume that $Q$ contains a double edge $uv$ such that $v$ is a leaf of $T_Q$. 
Let $a\in\f(u)$.
By Lemma~\ref{lem:distinct-free} there is an edge $vx$ colored with $a$. 
Consider the elementary move $\pi'$ obtained by uncoloring $vx$ and coloring $uv$ with $a$.
If $x\not\in V(Q)$, then $x\in \con_1(Q)$, hence by Lemma~\ref{lem:disjoint-free-W} we can put $U=\con_1(Q)$.
Otherwise, $Q(\pi')$ contains a cycle of length 3 or 4, so the claim follows by Case 2.1 and Observation~\ref{obs:niezmienne_ladunki}.

We are left with the case when every multiple edge of $Q$ is not incident with a leaf of $T_Q$.
Since $|V(Q)|=4$ it means that $T_Q$ is a path $uvpq$, and the edge $vp$ has multiplicity at least two in $Q$.
Let $a\in\f(p)$ and consider the edge $vx$ colored with $a$, which exists due to Lemma~\ref{lem:distinct-free}. 
We may move $Q$ by uncoloring $vx$ and coloring $vp$ with $a$. If $x\notin V(Q)$ then again $x\in \con_1(Q)$ and by Lemma~\ref{lem:disjoint-free-W} we can put $U=\con_1(Q)$. 
If $x\in V(Q)$ then either $x=u$ and we have obtained the case where the multiple edge is incident to a leaf of $T_Q$ or $x=q$ and we have obtained the case where $Q$ contains a cycle of length $3$; in both cases we get the claim by Observation~\ref{obs:niezmienne_ladunki}.
\end{proof}

\section{Collapsing subgraphs}\label{sec:collapsing}

\begin{definition}
Let $G$ be a multigraph with maximum degree $\Delta$ and $H$ an induced subgraph of $G$ such that $|V(H)|=3$. We say that $H$ is \emph{$k$-collapsible} if 
$|E(V(H),V(G)\setminus V(H))|\leq k$ and 
$|E(H\setminus x)|\geq |E(x,V(G)\setminus V(H))|$ for every vertex $x\in V(H)$
\end{definition}

If $H$ is a $k$-collapsible subgraph of $G$ then we may obtain a graph $G'$ by removing from $G$ all edges in $E(H)$ and identifying the three vertices in $V(H)$ into a single vertex $h$. Note that $G'$ has the maximum degree at most $\max\{k,\Delta(G)\}$. We say that $G'$ is obtained by collapsing $H$ to $h$ in $G$.

\begin{lemma}\label{lem_collapse}
Let $G'$ be obtained by collapsing a $k$-collapsible subgraph $H$ to $h$ in a multigraph $G$ with maximum degree $\Delta$.
\begin{enumerate}[$(i)$]
\item If $G$ does not contain an induced subgraph on three vertices with more than $\Delta+\lfloor\frac{k}{2}\rfloor$ edges then $G'$ does not contain an induced subgraph on three vertices with more than $\Delta+\lfloor\frac{k}{2}\rfloor$ edges.
\item Given a partial coloring $\pi '$ of $G'$ which colors at least $p|E(G')|$ edges we can construct a partial coloring $\pi$ of $G$ which colors at least $\min\{p,\frac{\Delta}{|E(H)|}\}|E(G)|$ edges.
\end{enumerate}
\end{lemma}

\begin{proof}
\begin{enumerate}
\item Every induced subgraph of $G'$ that does not contain $h$ is isomorphic to an induced subgraph of $G$ hence we only need to consider subgraphs of $G'$ containing $h$. However a subgraph of $G'$ containing $h$ and two other vertices can have at most $\Delta+\floor{\frac{\deg_{G'}(h)}{2}}\leq\Delta+\lfloor\frac{k}{2}\rfloor$ edges.
\item The edges in $G'$ correspond to edges in $G$ that are not in $E(H)$. Thus we may use the partial coloring $\pi '$ to obtain a partial coloring of $E(G)\setminus E(H)$. The edges in $G'$ incident with $h$ correspond to edges in $E(V(H),V(G)\setminus V(H))$. For a vertex $x\in V(H)$ we use the colors of edges in $E(x,V(G)\setminus V(H))$ to color edges in $E(H\setminus x)$. We use the remaining colors to color arbitrary edges of $E(H)$. We obtained a partial coloring $\pi$ such that $\min\{\Delta,|E(H)|\}$ edges of $E(H)$ are colored. Thus the claim follows.
\end{enumerate}
\end{proof}

\begin{lemma}\label{lem:main4and5}
Let $G$ be a connected multigraph of maximum degree $\Delta$. Let $\pi$ be a coloring maximizing $\Psi$.
\begin{enumerate}[$(i)$]
 \item If $\Delta=4$ and $G\neq 2K_3$ then $\pi$ colors at least $\frac{4}{5}|E|$ edges.
 \item If $\Delta=5$ and $G$ does not contain a $3$-collapsible subgraph then $\pi$ colors at least $\frac{5}{6}|E|$ edges.
\end{enumerate}
\end{lemma}
\begin{proof}
We will show that in both cases each free component $P$ of $\pi$ receives at least $|E(P)|\Delta$ units of charge, which gives the claim.

If $|E(P)|=1$ then we are done by Lemma~\ref{P1}. Thus by Corollary~\ref{cor:num-uncolored-0} we may assume that $|E(P)|=2$. We will first consider two special cases and then show that all other cases can be reduced via Observation~\ref{niezmienne_ladunki} to those two cases.

If $P$ is a path of length two and its edges are single edges in $G$ then we are done by Lemma~\ref{single_delta45}.

If $P$ consists of a double edge then we are done either by Lemma \ref{multiple_delta} or by Lemma \ref{multiple_delta-1}.

Assume that $P$ is a path $xyz$ and there are at least two $xy$ edges in $G$. Let $a$ denote the color of a colored $xy$ edge. If $a\in\overline{\pi}(z)$ then we uncolor the $a$-colored $xy$ edge and color the $yz$ edge with $a$. We obtain a component consisting of a double edge and we are done by Observation~\ref{niezmienne_ladunki}. 
So assume that $a\notin\overline{\pi}(z)$. Since $a\not\in\f(x)\cup\f(y)$, we have $a\not\in \f(P)$.
By Lemma~\ref{lem:distinct-free} this means that $\Delta=5$ and $\overline{\pi}(x)=\{b\}$, $\overline{\pi}(z)=\{c\}$ for some distinct colors $b$ and $c$. Consider the edges $yw_b$ and $yw_c$ colored respectively with $b$ and $c$. If $w_b=z$ or $w_c=x$ then we can uncolor this edge, color respectively $xy$ or $yz$ and obtain a component consisting of a double edge; again we are done by Observation~\ref{niezmienne_ladunki}. 
Otherwise, $\{x,y\}\cap\{w_b,w_c\}=\emptyset$ and we may move $P$ to a component with edges $yw_b$ and $yw_c$. 
Note that if $w_b\ne w_c$ then both edges $yw_b$ and $yw_c$ are single in $G$, because $\deg(y)\le 5$. We obtain a new component consisting of a double edge or two single edges in $G$ and again we are done Observation~\ref{niezmienne_ladunki}.
\end{proof}

In the following lemma we will use the function $\rho$. Although the definition may seem artificial it will be justified in the proof.

Let $\rho(\Delta,k,t)=\min(\{\frac{7}{2}\}\cup\{\frac{3\Delta-\alpha}{2e}\ |\ e,\alpha,\beta\in\mathbb{Z},$ $e\geq 2,$ $\alpha\geq 2e,$ $\beta\geq 0,$ $\alpha+\beta\leq\Delta,$ $e+\Delta-\beta\leq t,$ $2\beta+\Delta-\alpha\geq k+1\})$. 

\begin{lemma}\label{lem:aux}
Let $(G,\pi)$ be a coloring maximizing $\Psi$. 
Let $\Delta\geq 6$ be the maximal degree of $G$ and let $\Delta\leq t\leq\floor{\frac{3\Delta}{2}}$ and $0\leq k\leq\Delta$ be integers such that $G$ does not contain a $k$-collapsible subgraph and $G$ does not contain an induced subgraph on three vertices with more than $t$ edges. 
Let $H$ be a set of three vertices of $G$ and let $Q$ be a free component of $(G,\pi)$ such that $|E(Q)|\ge 2$ and $V(Q)\subseteq H \subseteq \con_1(Q)$. 
If for every $v\in H$, for every color $a\in \f(v)$ the two vertices of $H\setminus\{v\}$ are connected by an edge colored with $a$, then
$\ch(Q)\geq\rho(\Delta,k,t)|E(Q)|$.
\end{lemma}

\begin{proof}
Let $\alpha=|\f(H)|$, by Corollary~\ref{cor:num-uncolored-0} we have $\alpha\geq |\f(Q)|\ge 2|E(Q)|$. 
By Lemma~\ref{lem:disjoint-free-W}, for every color $a\in\f(H)$ there are exactly two vertices of $H$ incident with an edge colored with $a$, so by Observation~\ref{obs:1/2} $Q$ receives the charge of at least $\alpha$ from edges colored with colors in $\f(H)$.
Clearly,  for every color $a\not\in\f(H)$ there are exactly three vertices of $H$ incident with an edge colored with $a$,
and hence $Q$ receives the charge of at least $\frac{3}{2}(\Delta-\alpha)$ charge from edges colored with colors not in $\f(H)$. 
Thus $\ch(Q)\geq\frac{1}{2}(3\Delta-\alpha)$. In what follows, we show that $\frac{1}{2}(3\Delta-\alpha)\ge \rho(\Delta,k,t)|E(Q)|$, that is we will define a $\beta$ such that for $e=|E(Q)|$ either $\ch(Q)\geq\frac{7}{2}$ or all the inequalities in definition of $\rho$ are satisfied.

Note that for any color $b\notin\f(H)$ there are either two or three distinct edges colored with $b$ and incident with a vertex of $H$. 
Let $C_2$ and $C_3$ be the sets of colors in $\{1,\ldots,\Delta\}\setminus\f(H)$ such that there are respectively two or three distinct edges colored that color and incident with a vertex of $H$. 
Let $\beta=|C_3|$, then $|C_2|=\Delta-\alpha-\beta$. 
We will show that if $2\beta+\Delta-\alpha\leq k$ and $\beta\leq|E(Q)|$ then the subgraph of $G$ induced by $H$ is $k$-collapsible. 
Indeed, for $b\in C_3$ all three edges in $I[H]$ colored with $b$ belong to $E(H,V(G)\setminus H)$. 
For $b\in C_2$ one edge belongs to $E(H,V(G)\setminus H)$ and one has both endpoints in $H$. 
Note that in $E(H,V(G)\setminus H)$ there are no edges colored with colors in $\f(H)$ so $|E(H,V(G)\setminus H)|=3\beta+(\Delta-\alpha-\beta)\leq k$. 
For every vertex $x\in H$ and color $b\in C_2$ if there is an edge in $E(x,V(G)\setminus H)$ colored with $b$ then there is also an edge in $E(H\setminus x,H\setminus x)$ colored with $b$. 
Recall that no edges colored with a color from $\f(H)$ leave $H$. 
Hence, to show that $|E(G[H]\setminus x)|\ge|E(x,V(G)\setminus H)|$ it suffices to prove that each pair of vertices in $H$ is connected by at least $\beta$ edges in $\pi^{-1}(\f(H)\cup\{\bot\})$. 
However for each $x\in V(Q)$ we have $\deg_Q(x)\leq|\f(x)|$ so we may assign to each edge of $Q$ two colors from $\f(Q)$ using one free color from each of its endpoints. By Lemma~\ref{lem:distinct-free} we may assign the colors so that every color is assigned at most once. 
Consider an edge $e\in E(Q)$, say $e=uv$, and let $p=H\setminus\{u,v\}$. 
Let $c_1\in\f(u)$ and $c_2\in\f(v)$ be the colors assigned to $e$, then by the assumption of the lemma there is an $up$ edge colored with $c_2$ and an $vp$ edge colored with $c_1$. 
It follows that every edge $e\in E(Q)$ forms a $K_3$ together with the edges in $E(G[H])$ colored with the assigned colors. 
Thus each pair of vertices in $H$ is connected by at least $|E(Q)|$ edges in $\pi^{-1}(\f(Q)\cup\{\bot\})$. 
Hence if $2\beta+\Delta-\alpha\leq k$ and $\beta\leq|E(Q)|$ then the subgraph of $G$ induced on $H$ is $k$-collapsible, contradicting the assumption. 
If $\beta>|E(Q)|$ then $\ch(Q)\geq\frac{1}{2}(3\Delta-\alpha)\geq\frac{1}{2}(2\alpha+3\beta)>\frac{1}{2}(4|E(Q)|+3|E(Q)|)=\frac{7}{2}|E(Q)|$ and the claim follows. Suppose $2\beta+\Delta-\alpha>k$. Note that $G[H]$ contains $\alpha$ edges colored with colors in $\f(H)$, $\Delta-\alpha-\beta$ edges colored with colors in $C_2$ and $e$ uncolored edges. Thus $t\geq \alpha+\Delta-\alpha-\beta+e=e+\Delta-\beta$, so by definition $\ch(Q)\geq\rho(\Delta,k,t)|E(Q)|$ and the claim follows.
\end{proof}

\begin{lemma}\label{lem:7_9for3}
Let $(G,\pi)$ be a coloring maximizing $\Psi$. Let $Q$ be a free component of $(G,\pi)$ such that $|V(Q)|\in\{2,3\}$. Let $\Delta\geq 6$ be the maximal degree of $G$ and let $\Delta\leq t\leq\floor{\frac{3\Delta}{2}}$ and $0\leq k\leq\Delta$ be integers such that $G$ does not contain a $k$-collapsible subgraph and $G$ does not contain an induced subgraph on three vertices with more than $t$ edges. Then $\ch(Q)\geq\rho(\Delta,k,t)|E(Q)|$.
\end{lemma}
\begin{proof}
\mycase{1.} $|V(Q)|=3$ and $|E(Q)|\geq 3$. Let $V(Q)=\{u,v,p\}$. Note that $Q$ contains a cycle or a multiple edge.

\mycase{1.1.} $Q$ contains a cycle of length 3, i.e.\ the cycle $uvpu$. 
Pick any $a\in\f(Q)$, by symmetry assume that $a\in\f(u)$. By Lemma~\ref{lem:distinct-free} there is an $vx$ edge colored with $a$. 
If $x\neq p$ then we can move $Q$ by uncoloring $vx$ and coloring $uv$ with $a$. We obtain a component with four vertices and the claim follows by Lemma~\ref{lem:4_5for4} and Observation~\ref{obs:niezmienne_ladunki}. 
Thus we can assume that for every color in $\f(Q)$ there is an edge incident with two vertices in $V(Q)$ colored with that color. 
The claim follows from Lemma~\ref{lem:aux} applied for $H=V(Q)$.

\mycase{1.2.} $Q$ contains a multiple edge.
By symmetry assume that $uv$ is a multiple edge of $Q$ and  $vp$ is an edge of $Q$. Let $a\in\f(v)$, by Lemma~\ref{lem:distinct-free} there is an $ux$ edge colored with $a$. We may move $Q$ by uncoloring $ux$ and coloring $uv$ with $a$. If $x\neq p$ then we obtain a component with four vertices and the claim follows by Lemma \ref{lem:4_5for4} and Observation \ref{niezmienne_ladunki}. 
If $x=p$ then we obtain a component with three vertices and a cycle of length $3$. Hence the claim follows by Observation \ref{niezmienne_ladunki} and Case 1.1.

\mycase{2.} $|V(Q)|=3$ and $|E(Q)|=2$. 
If $\Delta\geq 7$ then by Lemma~\ref{lem:distinct-free} and Observation~\ref{obs:1/2} we have $\ch(Q)\geq\Delta\frac{|V(Q)|-1}{2}\geq 7=\frac{7}{2}|E(Q)|$. 
So we can assume that $\Delta=6$. Note that if there is a move $\pi '$ of $\pi$ such that $\con_1(Q(\pi '))\geq 4$ then by Observation~\ref{niezmienne_ladunki}, Lemma~\ref{lem:disjoint-free-W} and Observation~\ref{obs:1/2} we have 
$\ch(Q)=\ch(Q(\pi '))\geq \Delta \tfrac{|\con_1(Q(\pi'))|-1}{2} = 6\frac{4-1}{2}=\frac{9}{2}|E(Q)|$. 
Thus we assume that $\con_1(Q(\pi '))\le 3$ for every move $\pi '$ of $\pi$. Let $E(Q)=\{uv,vp\}$ and $c_u\in\f(u)$. By Lemma~\ref{lem:distinct-free} there is an $vx$ edge colored with $c_u$. By uncoloring $vx$ and coloring $uv$ with $c_u$ we obtain an elementary move $\pi_1$ of $\pi$. Hence $x\in \con_1(Q)$, so $x=p$. Thus for every color $c_u\in\f(u)$ there is an $vp$ edge colored with $c_u$ and therefore for every $c_u\in\overline{\pi_1}(u)$ there is an $vp$ edge colored with $c_u$ in $\pi_1$. Let $c_v\in\overline{\pi_1}(v)$ and 
consider the $py$ edge colored in $\pi_1$ with $c_v$. We may move $Q(\pi_1)$ by uncoloring $py$ and coloring $vp$ with $c_v$. Hence $y\in \con_1(Q(\pi_1))=\{u,v,p\}$, so $y=u$. Similarly for every $c_p\in\overline{\pi_1}(p)$ there is an $uv$ edge colored in $\pi_1$ with $c_p$. Thus we have obtained that for every color in $\overline{\pi_1}(\{u,v,p\})$ there is an edge incident with two vertices in $\{u,v,p\}$ colored with that color. The claim follows from Observation \ref{niezmienne_ladunki} and Lemma~\ref{lem:aux} applied for $(G,\pi_1)$ and $H=\{u,v,p\}$.

\mycase{3.} $|V(Q)|=2$. 
If $|E(Q)|=1$ then by Lemma~\ref{P1} we have $\ch(Q)\geq\Delta\geq 4|E(Q)|$. Suppose $|E(Q)|\geq 2$, let $V(Q)=\{u,v\}$ and $a\in\f(u)$. We may move $Q$ by uncoloring the $vx$ edge colored with $a$ and coloring a $uv$ edge with $a$. 
We obtain a free component $Q'$ with three vertices and are done by Case 2 and Observation \ref{niezmienne_ladunki}.
\end{proof}

\section{Proof of the main results}\label{sec:proof}

Now we are ready to describe the algorithm used to find the colorings from Theorem~\ref{thm:main}, Theorem~\ref{thm:main_1} and Theorem~\ref{thm:main_2} in polynomial time. 

\begin{algorithm}\label{alg}
Let $G$ be a multigraph with maximal degree $\Delta\geq 4$. Let $t$ be an integer such that $\floor{\frac{3\Delta}{2}}\geq t\geq 0$ and $G$ does not contain a subgraph with $3$ vertices and more than $t$ edges.
\begin{enumerate}
\item Let $k=\min\{\Delta,2(t-\Delta)+1\}$, $i:=0$ and $G_0:=G$. 

\item\label{alg:collapse}
While $G_i$ contains a $k$-collapsible subgraph $H$ let $i:=i+1$ and $G_i$ be the graph obtained by collapsing $H$ in $G_{i-1}$. By Lemma~\ref{lem_collapse} $G_i$ has maximal degree at most $\Delta$ and does not contain a subgraph with $3$ vertices and more then $t$ edges. 
The resulting graph $G_i$ does not contain a $k$-collapsible subgraph.
\item Let $\pi$ be the empty partial coloring of $G_i$, i.e.\ $\pi(E(G_i))=\{\bot\}$.
\item \label{alg:squeezing-components}As long as $(G_i,\pi)$ contains a free component with more than $\lfloor \Delta/2 \rfloor$ edges, use the procedure described in the proof of Lemma~\ref{lem:distinct-free} to find a new partial coloring $\pi'$, with increased number of colored edges and replace $\pi$ by $\pi'$.


\item\label{alg:while} For every free component $P$ determine $O(1)$ colored edges which send charge to $P$ so that $P$ gets the required amount; the edges are specified in proofs of relevant lemmas in Subsection~\ref{sec:charge}.

\item If in Step~~\ref{alg:while} charge is claimed from an uncolored edge or more than 1 unit of charge is claimed from a colored edge proceed as follows. 
The proofs of lemmas from Subsection~\ref{sec:charge} provide a coloring $\pi '$ with $\Psi(\pi ')>\Psi(\pi)$. 
Replace $\pi$ with $\pi '$ and repeat step~\ref{alg:while}. 
\item For $j=i,\ldots,1$ use Lemma~\ref{lem_collapse} to obtain a coloring of $G_{i-1}$ from a coloring of $G_i$. 
\end{enumerate}
\end{algorithm}

Now let us argue that the algorithm above takes only polynomial time.
Let $n,m$ denote the number of vertices and edges of the input graph $G$.
Finding a $k$-collapsible subgraph can be easily done in $O(n^3)$ time.
Since after collapsing such a subgraph the number of vertices decreases, Step~\ref{alg:collapse} takes $O(n^4)$ time.
One can easily check that the procedure from the proof of Lemma~\ref{lem:distinct-free} used to find the coloring $\pi'$ in Step~\ref{alg:squeezing-components} takes $O(nm)$ time. Since after each such operation the number of colored edges increases, Step~\ref{alg:squeezing-components} takes $O(nm^2)$ time.
Now we focus on Step~\ref{alg:while}. Note that although in principle a component $P$ can receive charge from many edges, which may be hard to find, in the proofs in Section~\ref{sec:charge} we needed only charge from edges incident to $P$ or to a component obtained from $P$ by at most two elementary moves, and the moves were always specified. Since every such edge sends at least $1/2$ unit of charge and every component needs only $O(1)$ units of charge per edge, to each component $P$ we assign $O(|E(P)|)=O(\Delta)$ edges and it is easy to find these edges in $O(m)$ time.
Hence one execution of Step~\ref{alg:while} takes only $O(mn)$ time. There are $O(mn^{\Delta/2})$ possible values of the potential $\Psi$, so the total time spent on Step~\ref{alg:while} is $O(m^2n^{\Delta/2+1})$. The last step takes time linear in the size of all the $k$-collapsible subgraphs found, which is bounded by $O(m)$. We conclude that the algorithm takes $O(m^2n^{\Delta/2+1})$ time, which is polynomial for every fixed value of $\Delta$.

Now we will prove that if $G$ satisfies certain conditions then Algorithm~\ref{alg} constructs a coloring which colors sufficiently many edges.
It will be convenient to begin with the proof of Theorem~\ref{thm:main_1}.

\begin{proof}[Proof of Theorem \ref{thm:main_1}]
Let $G$ be a multigraph of maximum degree $\Delta$. 
Fix an integer $t$ such that $\floor{\frac{3\Delta}{2}}\geq t\geq\left(\frac{1}{2}\sqrt{22}-1\right)\Delta$ and $G$ does not contain a subgraph with $3$ vertices and more than $t$ edges. 
Let $\pi$ be the coloring constructed for $G$ and $t$ by Algorithm~\ref{alg}. 
Let $k=\min\{\Delta,2(t-\Delta)+1\}$. 
Assume $G$ contains a $k$-collapsible subgraph $H$. Let $G'$ be the graph obtain by collapsing $H$.
Since $\Delta+\lfloor k/2\rfloor=t$, Lemma~\ref{lem_collapse}(i) tells us that $G'$ does not contain a three vertex subgraph with more than $t$ edges. Moreover, Lemma~\ref{lem_collapse}(ii) states that if we show $G'$ has a $\Delta$-edge-colorable subgraph with at least $\tfrac{\Delta}t|E(G')|$ edges then $G$ contains a $\Delta$-edge-colorable subgraph with at least $\tfrac{\Delta}t|E(G)|$ edges. 
Hence in what follows we can assume that $G$ does not contain $k$-collapsible subgraphs.

If $\Delta\le 5$ then $\floor{\frac{3\Delta}{2}}-1<\left(\frac{1}{2}\sqrt{22}-1\right)\Delta$ so $t=\floor{\frac{3\Delta}{2}}$. 
Then the claim follows from Shannon's theorem. 

Assume that $\Delta\geq 6$. 
We will show that every free component $P$ of $(G,\pi)$ receives at least $\frac{\Delta}{t-\Delta}|E(P)|$ charge. It follows that $\pi$ colors at least $\frac{\Delta}{t-\Delta}|E(G)|/(\frac{\Delta}{t-\Delta}+1)=\frac{\Delta}{t}|E(G)|$ edges, as required. Note that $t\geq\left(\frac{1}{2}\sqrt{22}-1\right)\Delta$ implies $\frac{\Delta}{t-\Delta}\leq\frac{\sqrt{22}+4}{3}\approx 2.9$.

Let $P$ be a free component of $(G,\pi)$. If $|V(P)|\geq 4$ then by Corollary~\ref{cor:4_5for5+} and Lemma~\ref{lem:4_5for4} we have $\ch(Q)\geq 4|E(Q)|$. If $|V(P)|\leq 3$ then by Lemma~\ref{lem:7_9for3} we have $\ch(P)\geq\rho(\Delta,k,t)|E(P)|$. Thus it remains to show that $\rho(\Delta,k,t)\geq\frac{\Delta}{t-\Delta}$.

Let $\alpha$, $\beta$ and $e$ be as in definition of $\rho(\Delta,k,t)$. 
If $k=\Delta$ then $2\beta>\alpha$ so $6\Delta\geq 6\alpha+6\beta>9\alpha$. 
Consequently $\frac{3\Delta-\alpha}{2e}=\frac{6\Delta-2\alpha}{4e}>\frac{9\alpha-2\alpha}{2\alpha}=\tfrac{7}2$ and thus $\rho(\Delta,k,t)\geq\frac{\Delta}{t-\Delta}$ as required. 
So assume $k=2(t-\Delta)+1$. 
Then $2\beta+\Delta-\alpha>2t-2\Delta$ so $2\beta+3\Delta>2t+\alpha$. 
It follows that $9\Delta-3\alpha=7\Delta-3\alpha+2\Delta\geq 7\Delta-3\alpha+2\alpha+2\beta=4\Delta-\alpha+(2\beta+3\Delta)>4\Delta-\alpha+2t+\alpha=4\Delta+2t$. 
In particular, this gives $3\alpha < 5\Delta - 2t$, so $6e\leq 3\alpha<5\Delta-2t$. 
Thus $\frac{3\Delta-\alpha}{2e}=\frac{9\Delta-3\alpha}{6e}>\frac{4\Delta+2t}{5\Delta-2t}$. 
It is easy to verify that for $\frac{3}{2}\geq \frac{t}{\Delta}\geq\frac{1}{2}\sqrt{22}-1$ we have $\frac{4\Delta+2t}{5\Delta-2t}\geq\frac{\Delta}{t-\Delta}$. Thus the claim follows.
\end{proof}

\begin{proof}[Proof of Theorem \ref{thm:main_2}]
Let $G$ be a multigraph of maximum degree $\Delta$. 
Assume $\Delta$ is even, $\Delta\ge 4$ and $G$ does not contain $\frac{\Delta}{2}K_3$ as a subgraph. 
This means that $G$ does not contain a three vertex subgraph with more than $t=\tfrac{3\Delta}2-1$ edges.
We will show that $G$ contains a $\Delta$-edge-colorable subgraph with at least $\tfrac{\Delta}t=\Delta/(\lfloor\tfrac{3\Delta}2\rfloor-1)$ edges.
If $\Delta\geq 8$ then $\frac{3\Delta}{2}-1\geq\left(\frac{1}{2}\sqrt{22}-1\right)\Delta$ and the claim follows by Theorem~\ref{thm:main_1}. 
If $\Delta=4$ then the claim follows by Lemma~\ref{lem:main4and5}. 
Hence we are left with the case $\Delta=6$. 
Then $t=8$; let $k=\min\{\Delta,2(t-\Delta)+1\}=5$.
By the argument from the beginning of the proof of Theorem~\ref{thm:main_1} we obtain that it is sufficient to show that $\rho(6,5,8)\geq\tfrac{\Delta}{t-\Delta}=3$. 
Let $\alpha$, $\beta$ and $e$ be as in definition of $\rho(\Delta,k,t)$. We have $2\beta+6-\alpha\geq 6$ and $6\geq\alpha+\beta$ and thus $12\geq 2\alpha+2\beta\geq 3\alpha$. Hence $\alpha\leq 4$ and $\frac{3\Delta-\alpha}{2e}\geq \frac{18-4}{4}=\frac{7}{2}$. So $\rho(6,5,8)=\frac{7}{2}>3$ and the claim follows.

Now assume $\Delta$ is odd, $\Delta\ge 5$ and $G$ does not contain $\frac{\Delta-1}{2}K_3+e$ as a subgraph. 
This means that $G$ does not contain a three vertex subgraph with more than $t=\tfrac{3}2(\Delta-1)$ edges.
We will show that $G$ contains a $\Delta$-edge-colorable subgraph with at least $\tfrac{\Delta}t=\Delta/(\lfloor\tfrac{3\Delta}2\rfloor-1)$ edges.
Similarly as for even $\Delta$, if $\Delta\geq 11$ then $\frac{3}{2}(\Delta-1)\geq\left(\frac{1}{2}\sqrt{22}-1\right)\Delta$ and the claim follows by Theorem~\ref{thm:main_1}. 
If $\Delta=5$ then by Lemma~\ref{lem_collapse} we can assume that $G$ does not contain a 3-collapsible subgraph so the claim follows by Lemma~\ref{lem:main4and5}. 
Thus we only need to verify the claim for $\Delta=7$ and $\Delta=9$ and, as argued in the beginning of the proof of Theorem~\ref{thm:main_1}, this can be achieved by showing respectively $\rho(7,5,9)\geq\frac{\Delta}{t-\Delta}=\frac{7}{2}$ and $\rho(9,7,12)\geq\frac{\Delta}{t-\Delta}=3$. 
If $\Delta=7$ then $2\beta+7-\alpha\geq 6$ and $7\geq\alpha+\beta$ and thus $14\geq 2\alpha+2\beta\geq 3\alpha-1$. 
Hence $\alpha\leq 5$ and, since $\alpha\ge 2e$ and $e$ is integer, $e\leq 2$. 
It follows that $\frac{3\Delta-\alpha}{2e}\geq\frac{21-5}{4}=4$, so $\rho(7,5,9)=\frac{7}{2}$. If $\Delta=9$ then $2\beta+9-\alpha\geq 8$ and $9\geq\alpha+\beta$ and thus $18\geq 2\alpha+2\beta\geq 3\alpha-1$. Hence $\alpha\leq 6$ and $\frac{3\Delta-\alpha}{2e}\geq\frac{27-6}{6}=\frac{7}{2}$. We get $\rho(9,7,12)=\frac{7}{2}$, as required.
\end{proof}

\begin{proof}[Proof of Theorem \ref{thm:main}]
Let $G$ be a connected multigraph of maximum degree $\Delta$. 
For $\Delta=3$ the claim was proved by Kamiñski and Kowalik~\cite{swat,Bey}, so in what follows we assume $\Delta\ge 4$.
Suppose $\Delta$ is even. Since $G\ne \frac{\Delta}{2}K_3$ and $G$ is connected it follows that $G$ does not contain $\frac{\Delta}{2}K_3$ as a subgraph. 
Hence the claim follows from Theorem~\ref{thm:main_2}.

Assume $\Delta$ is odd. 
If $G$ consists of two copies of $\frac{\Delta-1}{2}K_3+e$ joined by an edge then we can easily color $2\Delta+1$ of the $3\Delta$ edges of $G$ and the claim follows. 
So assume that $G$  does not consist of two copies of $\frac{\Delta-1}{2}K_3+e$ joined by an edge. 
Let $\mathcal{H}$ be the set of subgraphs of $G$ isomorphic to $\frac{\Delta-1}{2}K_3+e$. 
Then each $H\in\mathcal{H}$ is a $2$-edge-connected component of $G$ joined with $G\setminus H$ by a single edge $e_H$. 
Let $G'=G\left[V(G)\setminus\bigcup_{H\in\mathcal{H}}V(H)\right]$. 
We can color at least $\frac{2\Delta}{3\Delta-3}|E(G')|$ edges of $G'$ by Theorem~\ref{thm:main_2}. 
The partial coloring of $G'$ gives a partial coloring of $G$. 
For every $H\in\mathcal{H}$ we can additionally color $e_H$ and $\Delta$ of the $\frac{3\Delta-1}{2}$ edges of $H$. 
Thus we have colored at least 
$\frac{2\Delta}{3\Delta-3}|E(G')|+\sum_{H\in\mathcal{H}}(\Delta+1) = 
 \frac{2\Delta}{3\Delta-3}|E(G')|+\sum_{H\in\mathcal{H}}\tfrac{\Delta+1}{|E(H)|+1}(|E(H)|+1) = 
 \frac{2\Delta}{3\Delta-3}|E(G')|+\sum_{H\in\mathcal{H}}\tfrac{2\Delta+2}{3\Delta+1}(|E(H)|+1) > \frac{2\Delta+1}{3\Delta}|E(G)|$ edges of $G$ and the claim follows.
\end{proof}

\section{Approximation Algorithms}
Following~\cite{AH96}, let $c_k(G)$ be the maximum number of edges of a $k$-edge-colorable subgraph of $G$. 
We use the following result of Kami\'nski and Kowalik.

\begin{theorem}[\cite{swat,Bey}]
\label{th:meta-algorithm}
Let $\mathcal{G}$ be a family of graphs and let $\mathcal{F}$ be a $k$-normal
family of graphs.
Assume there is a polynomial-time algorithm which for every $k$-matching $H$ of
a graph in $\mathcal{G}$, such that $H\not\in\mathcal{F}$ finds its $k$-edge
colorable subgraph with at least $\alpha |E(H)|$ edges.
Moreover, let \[\beta = \min_{A,B\in\mathcal{F}\atop\text{$A$ is not
$k$-regular}}\frac{c_k(A)+c_k(B)+1}{|E(A)|+|E(B)|+1}\text{\quad and\quad }
\gamma=\min_{A\in\mathcal{F}}\frac{c_k(A)+1}{|E(A)|+1}.\]
Then, there is an approximation algorithm for the maximum $k$-ECS problem  for graphs in $\mathcal{G}$ with approximation ratio $\min\{\alpha,\beta,\gamma\}$.
\end{theorem}

Since the definition of $k$-normal family is very technical, we refer the reader to \cite{Bey} for its definition. As a direct consequence of Theorem~\ref{thm:main} and Theorem~\ref{th:meta-algorithm} we get the following two results.

\begin{theorem}
Let $k\geq 4$ be even. The maximum $k$-ECS problem has a $\frac{2k+2}{3k+2}$-approximation algorithm for multigraphs.
\end{theorem}
\begin{proof}
Let $\mathcal{F} = \{\frac{k}{2}K_3\}.$ It is easy to check that $\mathcal{F}$ is $k$-normal.
Now we give the values of parameters $\alpha, \beta$ and $\gamma$ from Theorem \ref{th:meta-algorithm}.
By Theorem \ref{thm:main}, $\alpha=\frac{2k}{3k-2}.$
We have $c_k(\frac{k}{2}K_3) = k$ and $|E(\frac{k}{2}K_3)| = \frac{3k}{2}.$
Hence, $\beta = \infty$ and $\gamma = \frac{2k+2}{3k+2}$.
\end{proof}

\begin{theorem}
Let $k\geq 5$ be odd. The maximum $k$-ECS problem has a $\frac{2k+1}{3k}$-approximation algorithm for multigraphs.
\end{theorem}
\begin{proof}
Let $\mathcal{F} = \{\frac{k-1}{2}K_3 + e\}.$ It is easy to check that $\mathcal{F}$ is $k$-normal.
Now we give the values of parameters $\alpha, \beta, \gamma$ and $\delta$
for Theorem \ref{th:meta-algorithm}.
By Theorem \ref{thm:main}, $\alpha=\frac{2k}{3(k-1)}.$
We have $c_k(\frac{k-1}{2}K_3 + e) = k$ and $|E(\frac{k-1}{2}K_3 + e)| = \frac{3k-1}{2}.$
Hence, $\beta = \frac{2k+1}{3k}$ and $\gamma = \frac{2k+2}{3k+1}.$
\end{proof}

\section*{Acknowledgments}

The second and third author are partially supported by National Science Centre of Poland (grant N206 567140).
The work has been done while the first author was staying at University of Warsaw.

\bibliographystyle{abbrv}

\end{document}